\documentclass[11pt]{article}
\usepackage[letterpaper,margin=1in]{geometry}

\usepackage{amsmath,amssymb,amsthm,amsfonts,latexsym,bbm,xspace,graphicx,float,mathtools,braket,dirtytalk}
\usepackage{enumitem}
\usepackage{xfrac}

\usepackage{yhmath}
\usepackage{subcaption}
\usepackage{sidecap}
\sidecaptionvpos{figure}{c}
\usepackage[linesnumbered,ruled,vlined]{algorithm2e}
\SetKwInput{KwInput}{Input}                %
\SetKwInput{KwOutput}{Output}
\SetKwFor{RepTimes}{repeat}{times}{end}

\usepackage{arydshln}
\usepackage{nicematrix}

\usepackage[colorinlistoftodos]{todonotes}

\makeatletter
\renewcommand*\env@matrix[1][*\c@MaxMatrixCols c]{%
  \hskip -\arraycolsep
  \let\@ifnextchar\new@ifnextchar
  \array{#1}}
\makeatother

\usepackage[pagebackref,colorlinks,citecolor=blue,,linkcolor=magenta,bookmarks=true]{hyperref}
\usepackage[nameinlink,capitalize]{cleveref}
\crefname{algocf}{Algorithm}{Algorithms}
\Crefname{algocf}{Algorithm}{Algorithms}
\crefname{algocfline}{Line}{Lines}
\Crefname{algocfline}{Line}{Lines}

\renewcommand{\backref}[1]{}

\renewcommand{\backrefalt}[4]{%
\ifcase #1 %
\or
[p.\ #2]%
\else
[pp.\ #2]%
\fi}

\newtheorem{theorem}{Theorem}[section]
\newtheorem*{namedtheorem}{\theoremname}
\newcommand{\theoremname}{testing}

\newtheorem{lemma}[theorem]{Lemma}
\newtheorem{claim}[theorem]{Claim}

\newtheorem{fact}[theorem]{Fact}
\newtheorem{corollary}[theorem]{Corollary}

\theoremstyle{definition}
\newtheorem{definition}[theorem]{Definition}

\newcommand{\E}{\mathop{\bf E\/}}

\newcommand{\tr}{\mathrm{tr}}  
\newcommand{\poly}{\mathrm{poly}}

\newcommand{\polylog}{\mathrm{polylog}}

\newcommand{\dist}{\mathrm{dist}}

\newcommand{\R}{\mathbb R}
\newcommand{\C}{\mathbb C}

\newcommand{\Z}{\mathbb Z}
\newcommand{\F}{\mathbb F}

\newcommand{\eps}{\varepsilon}

\newcommand{\calZ}{\mathcal{Z}}

\newcommand{\weyl}{\mathrm{Weyl}}

\newcommand{\abs}[1]{\lvert #1 \rvert}

\newcommand{\norm}[1]{\lVert #1 \rVert}

\newcommand{\ketbra}[2]{\ket{#1}\!\!\bra{#2}}

\newcommand{\sympcomp}{\perp}

\newcommand{\ignore}[1]{}

\newcommand{\anote}[1]{}
\newcommand{\jnote}[1]{}
\newcommand{\hnote}[1]{}

\newcounter{termcounter}[equation]
\renewcommand{\thetermcounter}{\the\numexpr\value{equation}+1\relax.\roman{termcounter}}
\crefname{term}{term}{terms}
\creflabelformat{term}{#2\textup{(#1)}#3}

\makeatletter
\def\term{\@ifnextchar[\term@optarg\term@noarg}%
\def\term@optarg[#1]#2{%
  \textup{#1}%
  \def\@currentlabel{#1}%
  \def\cref@currentlabel{[][2147483647][]#1}%
  \cref@label[term]{#2}}
\def\term@noarg#1{%
  \refstepcounter{termcounter}%
  \textup{\thetermcounter}%
  \cref@label[term]{#1}}
\makeatother

\usepackage{enumitem}

\newcommand{\entropy}{\mathsf{S}}
\newcommand{\binentropy}{\mathsf{H}}

\title{Pseudoentanglement Ain't Cheap}

\author{Sabee Grewal\thanks{\texttt{sabee@cs.utexas.edu}. The University of Texas at Austin.} \and Vishnu Iyer\thanks{\texttt{vishnu.iyer@utexas.edu}. The University of Texas at Austin.}\and William Kretschmer\thanks{\texttt{kretsch@berkeley.edu}. Simons Institute for the Theory of Computing, UC Berkeley.} \and Daniel Liang\thanks{\texttt{dl88@rice.edu}. Rice University.}}

\date{}

\begin{document}

\maketitle

\begin{abstract}
We show that any pseudoentangled state ensemble with a gap of $t$ bits of entropy requires $\Omega(t)$ non-Clifford gates to prepare.
This bound is tight up to polylogarithmic factors if linear-time quantum-secure pseudorandom functions exist.  
Our result follows from a polynomial-time algorithm to estimate the entanglement entropy of a quantum state across any cut of qubits. %
When run on an $n$-qubit state that is stabilized by at least $2^{n-t}$ Pauli operators, our algorithm produces an estimate that is within an additive factor of $\frac{t}{2}$ bits of the true entanglement entropy.
\end{abstract}

\section{Introduction}

Recent work \cite{aaronson2024pseudoent} introduced the notion of \textit{pseudoentangled} quantum states, in analogy with pseudorandomness in classical computation. An ensemble of quantum states is said to be pseudoentangled if states in the ensemble have low entanglement across every bipartition, but they are difficult to distinguish from states with much larger entanglement. A more formal definition is the following:

\begin{definition}
    A pseudoentangled ensemble with gap $f(n)$ vs.\ $g(n)$ (where $f(n) > g(n)$) consists of two ensembles of $n$-qubit states $\{\ket{\Psi_k}, \ket{\Phi_k}\}_k$ indexed by a key $k \in \{0,1\}^{\poly(n)}$ such that
    \begin{itemize}
        \item $\ket{\Psi_k}$ and $\ket{\Phi_k}$ are preparable in quantum polynomial time.
        \item With probability at least $1 - \poly(n)$ over the choice of $k$, the entanglement entropy across every cut of size $\Omega(f(n))$ of $\ket{\Psi_k}$ (respectively, $\ket{\Phi_k}$) is $\Theta(f(n))$ (respectively, $\Theta(g(n))$).
        \item For any polynomial $p(n)$, no polynomial-time quantum adversary can distinguish 
        \[
        \rho \coloneqq \E_k \left[ \ketbra{\Psi_k}{\Psi_k}^{\otimes p(n)}\right] \qquad \text{and} \qquad \sigma = \E_k \left[ \ketbra{\Phi_k}{\Phi_k}^{\otimes p(n)}\right]
        \] with better than negligible success probability.
    \end{itemize}
\end{definition}

\cite{aaronson2024pseudoent} showed that pseudoentangled states can be instantiated in polynomial time and logarithmic depth, assuming the existence of quantum-secure one-way functions. So, under a standard cryptographic assumption, there exists an efficient construction of pseudoentanglement.
Nevertheless, in applications, we sometimes need constructions that are even simpler and more efficient, due to constraints beyond total gate complexity. 
For example, \cite{aaronson2024pseudoent} suggested that a construction of ``holographic'' pseudoentangled states might imply that the AdS/CFT dictionary is hard to compute. However, it remains open to build such pseudoentangled states that are compatible with the laws of AdS/CFT.  
There are also various other measures of quantum state complexity, beyond circuit depth. 
For example, a counting argument shows that stabilizer states and free-fermionic states can require super-logarithmic depth, but these states are also ``easy'' in the sense that their evolutions are efficiently classically simulable \cite{aaronson2004simulation,valiant2002quantum} and they are efficiently learnable \cite{aaronson43identifying,montanaro-bell-sampling,aaronson2023efficient}. 

In this work, we study the relationship between pseudoentanglement and non-Clifford complexity.
The Clifford group is a remarkably useful object in quantum information that consists of all quantum circuits generated by Hadamard, Phase, and CNOT gates.
Clifford gates are almost universal for quantum computing: the addition of any single-qubit non-Clifford gate gives rise to a universal gate set, as shown by Shi \cite{shi2002toffoli}. Generally speaking, Clifford gates are ``cheaper'' than non-Clifford gates, in a sense that can be formalized in a variety of applications.
Examples where the cost of a quantum operation is dominated by non-Clifford gates include quantum fault tolerance based on magic state distillation \cite{bravyi2005magicstates}, near-Clifford classical simulation algorithms \cite{aaronson2004simulation,BrayviPhysRevLett.116.250501, RallPhysRevA.99.062337, Bravyi2019simulationofquantum}, and quantum learning algorithms based on the stabilizer formalism \cite{lai2022learning,grewal2023efficient,grewal2023efficient2,grewal2023improved,grewal_et_al:LIPIcs.ITCS.2023.64,leone-stabilizer-nullity,leone2023learning,hangleiter2023bell,Chia2023}. 

A related work by Grewal, Iyer, Kretschmer, and Liang \cite{grewal_et_al:LIPIcs.ITCS.2023.64,grewal2023improved} investigated the stabilizer complexity of a different cryptographic object called \textit{pseudorandom} quantum states. 
These are ensembles of quantum states that cannot be distinguished from Haar-random by any polynomial-time adversary. The main result of \cite{grewal2023improved}, improving upon earlier work by the same authors, shows that $n$-qubit pseudorandom states require at least $\Omega(n)$ non-Clifford gates.
The present work asks whether a similar lower bound on non-Clifford resources holds for quantum pseudoentanglement.

As it happens, several of the known pseudoentangled state ensembles are \textit{also} pseudorandom states ensembles, including the only known instantiations of pseudoentanglement that achieve an optimal gap of $\Theta(n)$ vs.\ $\omega(\log n)$ \cite{aaronson2024pseudoent, giurgicatiron2023pseudorandomness}.
However, \cite{aaronson2024pseudoent} further observed that a pseudoentangled ensemble need not be pseudorandom, nor vice-versa. Hence, it is not clear that the computational resources needed to construct pseudoentangled states mirror those for pseudorandom states.

On the other hand, existing work has made clear that that \textit{some} amount of non-Cliffordness is needed to generate pseudoentangled states. Fattal, Cubitt, Yamamoto, Bravyi, and Chuang \cite{fattal2004entanglement} gave an efficient algorithm for computing the entanglement entropy across any bipartition of a \textit{stabilizer state} (i.e., a state preparable using Clifford gates only).
Combined with Montanaro's algorithm for learning an unknown stabilizer state \cite{montanaro-bell-sampling}, this implies that any ensemble $\{\ket{\Psi_k},\ket{\Phi_k}\}_k$ of stabilizer states cannot be pseudoentangled.

Our main result is an algorithm with a much stronger guarantee than the combination of \cite{fattal2004entanglement,montanaro-bell-sampling}. In short, our algorithm estimates the entanglement entropy of an unknown quantum state across any cut of qubits, where the accuracy of the estimate scales with the number of Pauli operators that stabilize the state (i.e., the number of Pauli operators for which the state is a $+1$-eigenvector).

\begin{theorem}[Informal version of \Cref{cor:main-without-d}]\label{thm:informal-algo}
    There is a polynomial-time quantum algorithm that takes as input
    \begin{enumerate}
        \item $O(n^3)$ copies of an $n$-qubit quantum state $\ket{\psi}$ that is 
        stabilized by at least $2^{n-t}$ Pauli operators, and
        \item A bipartition $A \sqcup B = [n]$ of qubits,
    \end{enumerate}
    and outputs an estimate of the entanglement entropy of $\ket{\psi}$ across the partition. The estimate is within an additive factor of $\frac{t}{2}$ bits of the true entanglement entropy, with high probability.
\end{theorem}

We remark that the set of states for which this is applicable is rather large and includes states prepared by a Clifford circuit with up to $\lfloor \frac{t}{2} \rfloor$ auxiliary single-qubit non-Clifford gates (such as the $T$-gate).
When the algorithm is run on a stabilizer state (i.e., when $t=0$), then our algorithm outputs the entanglement entropy across any cut of qubits \emph{exactly}, recovering the aforementioned result of  \cite{fattal2004entanglement,montanaro-bell-sampling}.
On the other hand, the bounds degrade as $t$ becomes too close to $n$.

As a straightforward consequence, we show that pseudoentangled state ensembles require a number of non-Clifford gates that scales linearly in the pseudoentanglement gap.

\begin{theorem}[Restatement of \cref{cor:pes-lowerbound}]\label{thm:informal-pes}
    Any family of Clifford circuits that produces a pseudoentangled ensemble $\{\ket{\Psi_k}, \ket{\Phi_k}\}_k$ with entropy gap $f(n) \text{ vs. } g(n)$ satisfying $f(n) - g(n) \ge t$ must use $\Omega(t)$ auxiliary non-Clifford single-qubit gates.
\end{theorem}

So, pseudoentangled ensembles with the optimal $\Theta(n)$ vs.\ $\omega(\log n)$ gap require a linear number of non-Clifford gates. Interestingly, this matches the lower bound on non-Clifford gates needed for pseudorandom states \cite{grewal2023improved}.

\cref{cor:pes-lowerbound} is optimal up to polylogarithmic factors under plausible computational assumptions.
In particular, Ma \cite{mapersonal} constructs pseudoentangled state ensembles in $O(n \  \polylog(n))$ time under the assumption that linear-time quantum-secure pseudorandom functions exist.\footnote{The details will be included in a forthcoming work due to Ma \cite{mapersonal}.} 
This is to say that under this assumption, pseudoentangled state ensembles on $n$ qubits require at most $O(n \  \polylog(n))$ non-Clifford gates in total.
It is widely conjectured that linear-time \emph{classically}-secure pseudorandom functions exist \cite{ishai_10.1145/1374376.1374438, fan_10.1145/3519935.3520010}, and it is plausible that these (or other) constructions are also secure against quantum adversaries. 

Prior to this work, it was unknown if even $O(1)$ non-Clifford gates sufficed to construct pseudoentangled states. 
While one can efficiently learn states prepared by $O(\log n)$ non-Clifford gates \cite{grewal2023efficient,grewal2023efficient2}, it is unclear how to leverage those algorithms to estimate entanglement entropy. %
Following intuition from simulation algorithms for near-Clifford circuits, whose running times scale exponentially in the number of non-Clifford gates,
it was also conceivable that a super-logarithmic number of non-Clifford gates would be sufficient to construct pseudoentangled states.

\paragraph{Concurrent Work}
While finalizing this work, we became aware of independent and concurrent work by Gu, Oliveiro, and Leone \cite{gu2024magicinduced}.
Their \cite[Lemma 7]{gu2024magicinduced} resembles our \cref{thm:main-entanglement-bounds}. 
However, we only prove bounds for the von Neumann entropy, whereas \cite{gu2024magicinduced} prove similar bounds for any $\alpha$-R\'enyi entanglement entropy, which captures the von Neumann entropy as a special case. 

\subsection{Main Ideas}
For an $n$-qubit quantum state $\ket\psi$, let $\weyl(\ket\psi)$ denote the Pauli operators $P$ for which $P \ket\psi = 1$.\footnote{See \cref{def:unsigned-stabilizer-group} for a formal definition. We find it convenient to work with Weyl operators, a subset of Pauli operators that form a basis of $\C^{2^n \times 2^n}$. One can think of $\weyl(\ket\psi)$ as the stabilizer group of $\ket\psi$ with all of the phase information removed. For example, for any computational basis state $\ket{x}$, $\weyl(\ket x) \equiv \{I,Z\}^{\otimes n}$.} Let $A \sqcup B = [n]$ be a bipartition of qubits.
Denote $\weyl(\ket\psi)_A$ as the subset of Pauli operators in $\weyl(\ket\psi)$ that act only the qubits indexed by $A$, and define $\weyl(\ket\psi)_B$ analogously. 
We use $\dim(G)$ to refer to the minimum number of generators of a group $G$. 
It is easy to verify that $\weyl(\ket\psi)$,
$\weyl(\ket\psi)_A$, and $\weyl(\ket\psi)_B$ are abelian subgroups of the Pauli group.
We prove the following bounds on the entanglement entropy across $(A, B)$, which hold for any quantum state.

\begin{theorem}[Restatement of \cref{thm:main-entanglement-bounds}]\label{thm:informal-ent-bounds}
    Let $\rho = \ket{\psi}\!\!\bra{\psi}$ be an $n$-qubit quantum state and let $A \sqcup B = [n]$ be a partition of qubits. Then
    \[
        \dim\left(\weyl(\ket \psi)\right) - \dim\left(\weyl(\ket \psi)_B \right) - \abs{A}
        \le \entropy(\rho_A) \le \abs{A} - \dim\left(\weyl(\ket \psi)_A\right), 
    \]
    where $\rho_A = \tr_B(\rho)$ is the reduced density matrix of $\rho$ after tracing out the qubits in $B$ and $\entropy(\rho_A)$ is the entanglement entropy across $(A, B)$.
\end{theorem}

Crucially, the quantities in \cref{thm:informal-ent-bounds} can be (approximately) computed efficiently, given a polynomial number of copies of $\ket\psi$. 
First, we learn generators for $\weyl(\ket\psi)$ with Bell difference sampling, a quantum measurement that consumes four copies of a state and produces a Pauli operator (see the end of \cref{sec:symplectic-vector-spaces} for further detail).
Bell difference sampling many times, along with some classical post-processing, suffices to (approximately) learn generators of $\weyl(\ket\psi)$, as proven in prior work \cite{grewal2023improved,grewal2023efficient}.
Then, using these generators of $\weyl(\ket\psi)$, we can compute (approximations of) $\weyl(\ket\psi)_A$ and $\weyl(\ket\psi)_B$ in polynomial time by solving a system of linear constraints.
We note that these approximations suffice, due to a result of Audenaert \cite[Theorem 1]{audenaert2007sharp}, which relates the entanglement entropy of states that are close in trace distance.

Let us now explain how the upper and lower bounds in \cref{thm:informal-ent-bounds} are proved, and then explain some applications. 
At a high level, we argue that there exist Clifford circuits acting locally on either $A$ or $B$ that exhibit entanglement (or the lack thereof) in the system. Because the Clifford circuits are local, we conclude that the original state must have the same entanglement (or lack thereof). The formal proofs are given in \cref{sec:bounds}.

We begin with the upper bound.
Trivially, the entanglement entropy is at most $\abs{A}$. 
To simplify the presentation, define $a \coloneqq \dim(\weyl(\ket\psi)_A)$. 
By known techniques, one can construct a Clifford circuit acting only on $A$ that maps $\weyl(\ket\psi)_A$ to Pauli-$Z$ strings on a subset of $a$ qubits.\footnote{It is well known that Clifford unitaries map Pauli operators to Pauli operators, see \cref{sec:symplectic-vector-spaces} for more detail.}
This has the effect of mapping the state $\ket\psi$ to a \emph{product state} where a subset of $a$ qubits in $A$ are in a computational basis state and the remaining qubits are in some arbitrary state.   
The $a$ qubits cannot be entangled with the rest of the system, and our upper bound follows.

For the lower bound, we argue that there exist Clifford circuits acting locally on $A$ and $B$, respectively, that distill EPR pairs across qubits of $A$ and $B$.
Observe that the EPR state is stabilized by the Pauli operators generated by $X \otimes X$ and $Z \otimes Z$.
While these two Pauli operators commute with one another, they locally anticommute (i.e., $X$ and $Z$ do not commute).
Clifford circuits acting locally on $A$ (or $B$, respectively) do not affect the global or local commutation relations.
As such, any pair of locally anticommuting Pauli operators in $\weyl(\psi)$ can be mapped, via local Clifford circuits on $A$ and $B$, respectively, to $X \otimes X$ and $Z \otimes Z$, creating an EPR pair on one qubit of $A$ and $B$ each.
Our lower bound follows from counting the number of EPR pairs we can produce in this way.

For both the upper and lower bounds, the Clifford circuits can be found efficiently. 
Indeed, a similar approach played a crucial role in the tomography algorithms given in \cite{grewal2023efficient, grewal2023efficient2}. 
Additionally, one can view our lower bound as an efficient algorithm for entanglement distillation of quantum states with large stabilizer dimension (\cref{def:stabilizer-dimension}), which may be of independent interest.

We conclude by explaining how to apply \cref{thm:informal-ent-bounds} to get our entanglement estimation algorithm (\cref{thm:informal-algo}) and the pseudoentanglement lower bound (\cref{thm:informal-pes}). 
Our algorithm outputs upper and lower bounds $(u, \ell)$, essentially by computing the quantities appearing in \cref{thm:informal-ent-bounds}.
When run on a quantum state $\ket\psi$ that is stabilized by least $2^{n-k}$ Pauli operators, we prove that $u - \ell \leq k$. 
Therefore, $\frac{u + \ell}{2}$ will always be within an additive factor of $\frac{k}{2}$ bits of the true entanglement entropy.
As a corollary, we obtain a lower bound on preparing pseudoentangled states. %
Suppose we have two state ensembles $\{\ket{\Psi_k}\}_k$ and $\{\ket{\Phi_k}\}_k$ that are prepared with at most $t$ non-Clifford gates. 
If we run our entropy estimation algorithm on copies drawn from either ensemble, we will recover upper and lower bounds $(u, \ell)$ such that $u - \ell \leq 2t$.
Therefore, the pseudoentanglment gap between these ensembles is at most $2t$.

\section{Preliminaries}
For a positive integer $n$, $[n] \coloneqq \{1, 2, \dots, n\}$.
For $x = (a,b) \in \F_2^{2n}$, $a$ and $b$ always denote the first and last $n$ coordinates of $x$, respectively.
For vectors $v_1, \ldots, v_k$, $\langle v_1, \ldots, v_k \rangle$ denotes their span.
For matrix $X \in \mathbb{C}^{d \times d}$, $\norm{X}_1$ denotes the sum of the absolute values of its singular values (known as the trace norm, nuclear norm, or Schatten 1-norm). 
For quantum mixed states $\rho, \sigma$, $\mathrm{dist}_{\rm{tr}} \coloneqq \frac{1}{2}\norm{\rho - \sigma}_1$ is the trace distance.
For us, $\log$ denotes the logarithm with base $2$, and $\ln$ is the logarithm with base $e \approx 2.718$.

Let $A \sqcup B$ be a partition of $[n]$.
We refer to $(A, B)$ as a \emph{cut} of $n$ qubits.
Let $\rho$ be an $n$-qubit quantum state.
The entanglement entropy across $(A,B)$ is defined as 
\[
\entropy(\rho_A) \coloneqq - \tr(\rho_A \log \rho_A) = -\tr(\rho_B \log \rho_B), 
\]
where $\rho_A = \tr_B(\rho)$ and $\rho_B = \tr_A(\rho)$ are the states obtained by tracing out $B$ and $A$, respectively.

Define the binary entropy function $\binentropy(p)$ by
\[
\binentropy(p) = - p \log(p) - (1-p) \log(1-p).
\]
The following is a well-known upper bound on $\binentropy(p)$.
\begin{fact}\label{fact:binary-entropy-bound}
    \[
        \binentropy(p) \leq \left(4p(1-p)\right)^{1/\ln 4} \leq e \cdot p^{1/\ln 4} \leq e \cdot p^{0.72}.
    \]
\end{fact}

If two states are close in trace distance, then so is their entanglement entropy.
\begin{lemma}[Fannes-Audenaert inequality {\cite[Theorem 1]{audenaert2007sharp}}]
\label{lem:fannes}
    Let $\rho = \ketbra{\psi}{\psi}$ and $\sigma = \ketbra{\phi}{\phi}$ be $n$-qubit states satisfying $\dist_\tr(\ket\psi, \ket\phi) \le \eps$, and let $A \sqcup B = [n]$ be a partition. Then
    
    \[
    |\entropy(\rho_A) - \entropy(\sigma_A)| \le \eps n + \binentropy(\eps)
    \]
\end{lemma}

\subsection{Symplectic Vector Spaces and Weyl Operators}
\label{sec:symplectic-vector-spaces}
There is a deep connection between quantum information and symplectic vector spaces over $\F_2$ that we leverage throughout this work. 
Many in the quantum information and theoretical computer science communities may not be familiar with this connection,
so we take care to review these notions here. 

To obtain a symplectic vector space, one must equip a vector space with a symplectic form. 

\begin{definition}[Symplectic form]\label{def:symp-bilinear-form}
Let $V$ be a vector space over a field $\F$.
A symplectic form is a mapping $\omega : V \times V \to \F$ that satisfies the following conditions. 
\begin{enumerate}
    \item Bilinear: $\omega$ is linear in each argument separately.
    \item Alternating: For all $v \in V$, $\omega(v, v) = 0$. 
    \item Non-degenerate: If for all $v \in V$, $\omega(u, v) =0$, then $u = 0$.
\end{enumerate}
\end{definition}

A symplectic vector space is a pair $(V, \omega)$, where $V$ is a vector space and $\omega$ is a symplectic form. 
We will equip $\F_2^{2n}$ with the standard symplectic form, which we refer to as the $\emph{symplectic product}$. 

\begin{definition}[Symplectic product]
   For $x, y \in \F_2^{2n}$, we define the symplectic product as $[x,y] = x_1 \cdot y_{n+1} + x_2 \cdot y_{n+2} + \dots + x_n \cdot y_{2n} + x_{n+1} \cdot y_{1} + x_{n+2} \cdot y_{2} + \dots + x_{2n} \cdot y_{n}$.
\end{definition}
In this work, one should always view $\F_2^{2n}$ as a symplectic vector space equipped with the symplectic product.

The symplectic product allows us to define the symplectic complement. 

\begin{definition}[Symplectic complement]
   Let $T \subseteq F_2^{2n}$ be a subspace. 
   The symplectic complement of $T$, denoted by $T^\perp$, is defined by 
   \[
 T^{\perp} \coloneqq \{a \in \F_2^{2n} : \forall x \in T, [x, a] = 0\}.   
   \]
\end{definition}

We will also need the notion of isotropic and symplectic subspaces. 

\begin{definition}[Isotropic subspace]
A subspace $W \subseteq \F_2^{2n}$ is isotropic when $W \subseteq W^\perp$. Equivalently, $W$ is isotropic if and only if $[w_1, w_2] = 0$ for all $w_1, w_2 \in W$.
\end{definition}

\begin{definition}[Symplectic subspace]
A subspace $W \subseteq \F_2^{2n}$ is symplectic when $W \cap W^\perp = \{0\}$.
\end{definition}

Every symplectic space has a standard basis, which we refer to as the \emph{symplectic basis}.

\begin{fact}\label{fact:symp-basis}
    Any $2d$-dimensional symplectic space over $\F_2$ has a basis $\{x_1, \dots, x_d, z_1, \dots, z_d\}$ such that
    \[
        [x_i, z_j] = \delta_{ij} \qquad\text{and}\qquad [x_i, x_j] = [z_i, z_j] = 0.
    \]
    Any basis with the above form is referred to as a symplectic basis.
\end{fact}

The direct sum of two symplectic vector spaces is also symplectic.
This is a basic fact, but we include a proof for completeness.

\begin{fact}\label{fact:direct-sum-symp}
    If $V, W$ are symplectic vector spaces over $\F_2$, so is their direct sum. 
\end{fact}
\begin{proof}
    Denote the symplectic forms on $V$ and $W$ by $\omega_V$ and $\omega_W$, respectively.
    Let $A \coloneqq V \oplus W$, where $V \oplus W = \{(v,w) : v \in V, w \in W\}$. 
    Define the form $\omega_A$ on $A$ by $\omega_A(a_1, a_2)= \omega_V(v_1, v_2) + \omega_W(w_1, w_2)$ where $a_i = v_i + w_i$.
    We will prove that $(A, \omega_A)$ is symplectic.
    
    It is obvious that $A$ is a vector space, so it remains to prove that $\omega_A$ is a symplectic form. 
    Recall from \cref{def:symp-bilinear-form} that we must show that $\omega_A$ is bilinear, alternating, and non-degenerate.
    It is clear that $\omega_A$ is bilinear because it is the sum of two bilinear forms.
    For $a = (v,w) \in A$, we have $\omega_A(a, a) = \omega_V(v, v) + \omega_W(w, w) = 0$, so $\omega_A$ is alternating.
    
    Finally, we prove that $\omega_A$ is non-degenerate.
    Suppose we have an element $a = (v,w)$ such that for all $a' \in A$ we have $\omega_A(a, a')= \omega_V(v, v') + \omega_W(w, w')  = 0$.
    Now choose a $v' \in V$ such that $[v, v'] = 0$ (one must exist since $V$ is symplectic). 
    $\omega_A(a, a')= \omega_V(v, v') + \omega_W(w, w') = 0$ by assumption, and, because $\omega_V(v, v') = 0$, it follows that $\omega_W(w, w') = 0$ for all $w' \in W$. Since $\omega_W$ is non-degenerate, $w = 0$. A similar argument shows that $v = 0$. Therefore, $\omega_A$ is non-degenerate.
\end{proof}

Each element of $\F_2^{2n}$ can be identified with a \emph{Weyl operator}. 
For $x = (a,b) \in \F_2^{2n}$, let $a^\prime, b^\prime$ be the embeddings of $a, b$ into $\Z^n$, respectively. Then the Weyl operator $W_x$ is defined as 
\[
W_x \coloneqq i^{a^\prime \cdot b^\prime} (X^{a_1}Z^{a_1}) \otimes \dots \otimes (X^{a_n}Z^{b_n}).
\]
The symplectic structure of $\F_2^{2n}$ respects the commutation relations of the Weyl operators. Specifically, for $x,y \in \F_2^{2n}$, $[x,y] = 0$ iff $W_xW_y = W_y W_x$.
Therefore, working with symplectic vector spaces lets us discard cruft while retaining relevant algebraic structure.
We also note that Weyl operators form an orthogonal basis of $\C^{2^n \times 2^n}$ with respect to the Hilbert-Schmidt inner product, so every quantum state and unitary transformation can be written as a linear combination of Weyl operators. 

We define $\calZ \coloneqq 0^n \times \F_2^n$ as the subset of $\F_2^{2n}$ corresponding to Pauli-$Z$ strings.
We define the unsigned stabilizer group of a quantum state as the subspace of $\F_2^{2n}$ that stabilizes or anti-stabilizes the state.

\begin{definition}[Unsigned stabilizer group]\label{def:unsigned-stabilizer-group}
Given an $n$-qubit quantum state $\ket\psi$, $\weyl(\ket\psi) \coloneqq \{ x \in \F_2^{2n} : W_x \ket\psi = \pm \ket\psi\}$ is the unsigned stabilizer group of $\ket\psi$.
\end{definition}

It is easy to verify that $\weyl(\ket\psi)$ is an isotropic subspace.
We define the stabilizer dimension, which quantifies the size of the Pauli group stabilizing a given state.

\begin{definition}[Stabilizer dimension]\label{def:stabilizer-dimension}
    Let $\ket{\psi}$ be a $n$-qubit pure state. The \emph{stabilizer dimension} of $\ket{\psi}$ is the dimension of $\weyl(\ket{\psi})$ as a subspace of $\F_2^{2n}$.
\end{definition}

The stabilizer dimension of a state is closely related to the number of non-Clifford gates required to prepare it.

\begin{fact}[{\cite[Lemma 4.2]{grewal2023improved}}]
\label{lem:stab_dim}
    Let $\ket{\psi}$ be an $n$-qubit state which is the output of a Clifford circuit with at most $t$ single-qubit non-Clifford gates. Then $\ket{\psi}$ has stabilizer dimension at least $n-2t$.
\end{fact}

For a Clifford circuit $C$ and any $x \in \F_2^{2n}$, we define $C(x)$ to be the $y \in \F_2^{2n}$ such that $W_y = \pm CW_xC^\dagger$. We can extend this notation to subsets $S$ of $\F_2^{2n}$ by writing $C(S) = \{C(x) : x \in S\} $. Conjugation by any Clifford circuit is an automorphism of the Pauli group. Furthermore, $C(x)$ preserves the symplectic form.

\begin{fact}
  For any Clifford circuit $C$ and $x, y \in \F_2^{2n}$, $[C(x), C(y)] = [x,y]$.
\end{fact}
\begin{proof}
    Recall that $W_{C(x)}W_{C(y)} = (-1)^{[C(x), C(y)]}W_{C(y)}W_{C(x)}$. Suppose that $W_{C(x)} = (-1)^{c_1}CW_xC^\dagger$ and $W_{C(y)} = (-1)^{c_2}CW_yC^\dagger$. We have
    \begin{align*}
    W_{C(x)}W_{C(y)} &= (-1)^{c_1 + c_2} CW_xC^\dagger C W_y C^\dagger \\
    &= (-1)^{c_1 + c_2} CW_x W_y C^\dagger \\
    &= (-1)^{[x,y]} (-1)^{c_1 + c_2}    CW_y W_x C^\dagger \\
     &= (-1)^{[x,y]} (-1)^{c_1 + c_2}    CW_yC^\dagger C W_x C^\dagger \\
     &= (-1)^{[x,y]}  W_{C(y)} W_{C(x)}.
    \end{align*}
    Thus $[C(x), C(y)] = [x,y]$.
\end{proof}

Since the inverse of any Clifford circuit is itself a Clifford circuit, we have the following as a simple corollary:

\begin{corollary}
    Given a subspace $H \subseteq \F_2^{2n}$ and a Clifford circuit $C$, $H$ is isotropic (resp. symplectic) if and only if $C(H)$ is isotropic (resp. symplectic).
\end{corollary}

Finally, we remark on Bell difference sampling \cite{montanaro-bell-sampling, gross2021schur}, an algorithmic primitive used in this work. 
Define $p_\psi(x) \coloneqq 2^{-n} \braket{\psi|W_x|\psi}^2$.
Bell difference sampling is a quantum measurement that takes four copies of a state $\ket\psi$, and produces a sample $x \in \F_2^{2n}$ drawn from the distribution $q_\psi$ which is defined as 
\[
q_\psi(x) \coloneqq \sum_{a \in \F_2^{2n}} p_\psi(a) p_\psi(x+a).
\]
This process takes $O(n)$ time.
We refer readers to \cite[Section 2]{grewal2023improved} for further detail.

\section{Entanglement Entropy Bounds}\label{sec:bounds}

We prove upper and lower bounds on the entanglement entropy across any cut of qubits for any $n$-qubit quantum state $\ket\psi$. 
The quality of our bounds depends on $\weyl(\ket\psi)$. 
For example, if $\dim(\weyl(\ket \psi)) = 0$, our bounds become trivial, and, if $\dim(\weyl(\ket \psi)) = n$ (i.e., $\ket\psi$ is a stabilizer state), our bounds are tight, recovering the main result of \cite{fattal2004entanglement}.

To state our bounds, we must introduce some notation.
Let $A \subseteq [n]$ be a subset of qubits, and let $S$ be a subspace of $\F_2^{2n}$. We denote by $S_A$ the intersection of $S$ with operators that act only on qubits indexed by $A$. In symbols, we can express this as follows:
\[
S_A \coloneqq \{(x, z) \in S : \forall i \in [n] \setminus A, x_i = z_i = 0\}.
\]
So, for example, 
\[
\calZ_A \coloneqq \{(0^n, z) \in \F_2^{2n} : \forall i \in [n] \setminus A, z_i = 0\}
\]
are essentially the Pauli-$Z$ strings that act on the qubits indexed by $A$.

Computationally speaking, one can compute $S_A$ efficiently, given a basis of $S$, by solving a system of linear constraints to zero all coordinates corresponding to $i \in [n] \setminus A$.

In the remainder of this section, we prove the following theorem. 
\begin{theorem}
\label{thm:main-entanglement-bounds}
    Let $\rho = \ket{\psi}\!\!\bra{\psi}$ be an $n$-qubit quantum state and let $A \sqcup B = [n]$ be a partition of qubits. Then
    \[
        \dim\left(\weyl(\ket \psi)\right) - \dim\left(\weyl(\ket \psi)_B \right) - \abs{A}
        \le \entropy(\rho_A) \le \abs{A} - \dim\left(\weyl(\ket \psi)_A\right).
    \]
\end{theorem}

\subsection{Proof of Upper Bound}

\begin{lemma}
\label{lem:upper}
    Let $\rho = \ketbra{\psi}{\psi}$ be an $n$-qubit quantum state and let $A \subseteq [n]$ be a partition of qubits. Then
    \[
        \entropy(\rho_A) \le \abs{A} - \dim\left(\weyl(\ket \psi)_A\right).
    \]
\end{lemma}
\begin{proof}

    Let $A' \subseteq A$ be any set of size $\dim(\weyl(\ket{\psi})_A)$.
    By known techniques, one can find a Clifford circuit $C$ acting only on $A$ that maps the Paulis in $\weyl(\ket{\psi})_A$ to $\calZ_{A'}$.\footnote{An explicit algorithm for computing this $C$ can be found in \cite[Section 3]{grewal2023efficient}.} As a consequence, this $C$ behaves as $C\ket{\psi} = \ket{x}_{A'}\ket{\phi}_{[n] \setminus A'}$, where $\ket{x}_{A'}$ is a computational basis state on $A'$ and $\ket{\phi}_{[n] \setminus A'}$ is an arbitrary state on the remaining qubits. Because $C$ is local to $A$, it does not affect the entanglement entropy across the partition. Furthermore, it is clear that the qubits $\ket{x}_{A'}$ are unentangled from the rest of the system. As such, only the qubits in $A \setminus A'$ can exhibit entanglement across the partition, and there are $\abs{A} - \abs{A'} = \abs{A} - \dim\left(\weyl(\ket \psi)_A\right)$ many such qubits. So, the entanglement entropy across the partition is bounded above by $\abs{A} - \dim\left(\weyl(\ket \psi)_A\right)$.
\end{proof}

\subsection{Proof of Lower Bound}

\begin{lemma}\label{lemma:symplectic-subspace-inside}
Let $V \subseteq \F_2^{2n}$ be a symplectic subspace of dimension $2v$ and have $S \subseteq V$ be a subspace of dimension $v+k$. There exists a symplectic subspace $T \subseteq S$ with dimension at least $2k$.
\end{lemma}
\begin{proof}
    Take any nonzero $e_1 \in S$. Because $\dim(S) > v$, there exists some $f_1 \in S$ such that $[e_1, f_1] = 1$.
    Let $W_1$ be the span of $e_1$ and $f_1$.
    We will prove that $S = W_1 \oplus \left(W_1^\sympcomp \cap S\right)$ is a direct sum. 
    First, we argue that $W_1 \cap  \left(W_1^\sympcomp \cap S\right) = \{ 0\}$.\footnote{In fact, because $W_1$ is symplectic, even $W_1 \cap  W_1^\sympcomp = \{ 0\}$.}
    Take $z \in W_1$. Since $W$ is the span of $e_1$ and $f_1$, we can write $z = \alpha e_1 + \beta f_1$. 
    If $z$ is also in $W^\sympcomp \cap S \subseteq W^\sympcomp$, then $0 = [x, z] = \beta$ and $0 = [y, z] = \alpha$, so $z = 0$.
    Next, we prove that any $v \in S$ can be written as a sum of $w_1 \in W_1$ and $w_1^c \in W_1^\sympcomp \cap S$. 
    Clearly $w_1 \coloneqq [v, e_1]f_1 + [v, f_1]e_1 \in W_1$, and define $w_1^c \coloneqq v + [v, e_1]f_1 + [v, f_1]e_1$.
    It is easy to check $[e_1, w_1^c] = [f_1, w_1^c] = 0$, so $w_1^c \in W_1^\sympcomp$. 
    Furthermore, if since $e_1, f_1 \in W_1 \subset S$, $w_1^c \in S$ as well.
    It is clear that $v = w_1 + w_1^c$.

Repeat this process to collect pairs $(e_1, f_1), \dots, (e_r, f_r)$ until we have that $S \cap_{j=1}^{r}W_j^\sympcomp$ doesn't contain an $e_{r+1}$ and $f_{r+1}$ such that $[e_{r+1}, f_{r+1}] = 1$. 
Observe by linearity that any object in $x \in \cap_{j=1}^{r}W_j^\sympcomp$ must have $[x, y] = 0$ for all $y \in \bigoplus_{i=1}^r W_i$.
Consequently, $S \cap_{j=1}^{r}W_j^\sympcomp \oplus \langle e_1, \ldots, e_r\rangle$ is an isotropic subspace of dimension $v + k - r$. 
Since all isotropic subspaces within $V$ must have dimension at most $v$, 
\[
v+k - r \leq v \implies r \geq k.
\]
Hence, $\bigoplus_{i=1}^r W_i$ has dimension at least $2k$. Since each $W_i$ is symplectic, their direct sum is symplectic by \cref{fact:direct-sum-symp}. 
\end{proof}

\begin{lemma}
\label{lem:lower}
    Let $\rho = \ketbra{\psi}{\psi}$ be an $n$-qubit quantum state and let $A \sqcup B = [n]$ be a partition of qubits. Then
    \[
        \entropy(\rho_A) \ge 
        \dim\left(\weyl(\ket \psi)\right) - \dim\left(\weyl(\ket \psi)_B \right) - \abs{A}.
    \]
\end{lemma}
\begin{proof}
    Let $\{b_i\}_i$ be a basis for $\weyl(\ket\psi)_B$, and let $\{e_i\}_i$ be an extension such that, together, they span $\weyl(\ket\psi)$.
    Define the subspace $S \coloneqq \langle \{e_i\}_i \rangle$. Clearly $\dim(S) = \dim(\weyl(\ket\psi)) - \dim(\weyl(\ket\psi)_B)$ and $S \cap \left(\F_2^{2n}\right)_B = \{0 \}$.
    Define $e^A_i \in \left(\F_2^{2n}\right)_A$ to be $e_i$ except with the coordinates not in $A$ set to $0$.
    
    Define $S^A \coloneqq \langle \{e^A_i\}_i \rangle$.
    \vspace{-2em}
    \begin{changemargin}{1cm}{1cm} 
    \begin{claim}
        $\dim(S^A) = \dim(S) = \dim(\weyl(\ket\psi)) - \dim(\weyl(\ket\psi)_B)$
    \end{claim}
    \begin{proof}
    $\dim(S) \geq \dim(S^A)$ is trivial, so we just need to argue that the vectors $\{e_i^A\}_i$ are linearly independent.
    For the sake of contradiction, assume they're not, i.e., that there exists some set of indices $I \subseteq [\dim(S)]$ such that $\sum_{i \in I} e^A_i = 0$.
    Note that $\sum_{i \in I} e_i \neq 0$ because is $\{e_i\}_i$ is a basis by construction.
    Therefore, $\sum_{i \in I} e_i$ will be zero on the coordinates of $A$, but not $B$. That is, $\sum_{i \in I} e_i \in \weyl(\ket\psi)_B$, a contradiction.
    We conclude that $\dim(S^A) = \dim(\weyl(\ket\psi)) - \dim(\weyl(\ket\psi)_B)$.
    \end{proof}
    \end{changemargin}
    
    $(\F_2^{2n})_A$ is a symplectic subspace, and each $e^A_i \in (\F_2^{2n})_A$. Therefore, $S^A$ is a subspace of a $2\abs{A}$-dimensional symplectic subspace. 
    By \cref{lemma:symplectic-subspace-inside}, there must exist some symplectic subspace $T^A \subseteq S^A$ of dimension at least $2\left(\dim\left(S^A\right) - \abs{A}\right)  = 2\left(\dim\left(\weyl(\ket \psi)\right) - \dim\left(\weyl(\ket \psi)_B \right) - \abs{A}\right)$.
    Let $\{t_i^A\}_i$ be a symplectic basis of $T^A$. We can express each basis element as $t_i^A = \sum_j \alpha_j e^A_j$ for some setting of $\alpha_j \in \{0, 1\}$.
    Define $t_i \coloneqq \sum_j \alpha_j e_j$, and observe that their span defines a subspace $T \subseteq \F_2^{2n}$ that shares the same dimension as $T^\prime$.
    We can then similarly define $t^B_i \in \left(\F_2^{2n}\right)_B$ to be $t_i$ except with the coordinates not in $B$ set to $0$.
    Observe that $t_i = t^A_i + t^B_i$.
    By the linearity of the symplectic product and the fact that $\weyl(\ket \psi)$ is isotropic, \[0 = [t_i, t_j] = [t^A_i, t^A_j] + [t^B_i, t^B_j] \implies [t^A_i, t^A_j] = [t^B_i, t^B_j].\]
    Therefore, $\{t_i^B\}$ is also a symplectic basis for a symplectic subspace of $\left(\F_2^{2n}\right)_{B}$.

    Using a Clifford circuit $C^A$ acting locally on the qubits in $A$, we can perform the symplectic map that takes $\{t_i^A\}$ to the symplectic basis of $\left(\F_2^{2n}\right)_{A^\prime}$ where $A^\prime \subset A$ and $\abs{A^\prime} = \dim(T)$.\footnote{Again, this mapping is described in detail in \cite[Section 3]{grewal2023efficient}.}
    Note that $C^A(t_i^B) = t_i^B$.
    Using a second Clifford circuit $C^B$ acting locally on the qubits in $B$, we can take $\{t_i^B\}$ to the symplectic basis of $\left(\F_2^{2n}\right)_{B^\prime}$ where $B^\prime \subset B$ and $\abs{B^\prime} = \abs{A^\prime} = \dim(T)$.
    Note $C^B(C^A(t_i^A)) = C^A(t_i^A)$.
    
    \vspace{-2em}
    \begin{changemargin}{1cm}{1cm} 
    \begin{claim}
        $C^B(C^A(T))$ is a Lagrangian subspace of $\left(\F_2^{2n}\right)_{A^\prime \sqcup B^\prime}$.
    \end{claim}
    \begin{proof}
        It is clear from the actions of $C^A$ and $C^B$ that each $C^B(C^A(t_i))$ is a member of $\left(\F_2^{2n}\right)_{A^\prime \sqcup B^\prime}$.
        Furthermore, $\dim\left(C^B(C^A(T))\right) = \dim(T)$, which is half the dimension of $\left(\F_2^{2n}\right)_{A^\prime \sqcup B^\prime}$.
        Finally, since $C^B\left(C^A\left(\weyl(\ket \psi)\right)\right)$ is isotropic, so too must $C^B(C^A(T))$ as a subset of $C^B\left(C^A\left(\weyl(\ket \psi)\right)\right)$.
    \end{proof}
    \end{changemargin}

    We conclude that the state of the qubits indexed by $A' \sqcup B'$ is a stabilizer state $\ket\phi$ of $2\dim(T)$ qubits that is unentangled from the rest of the system.

    Our last step is to prove that $C^B(C^A(T))_{B^\prime} = \{0\}$, which will imply that the entanglement across $(A', B')$ is $\dim(T)/2$ by \cite[Eq. 1]{fattal2004entanglement}.
    First recall that $S \cap \left(\F_2^{2n}\right)_B = \{0\}$ by construction, which implies $T \cap \left(\F_2^{2n}\right)_B = \{0\}$ because $T\subseteq S$.
    Next, we note that $C^A$ has no effect on $\left(\F_2^{2n}\right)_B$ since its action is local to $A$. 
    Furthermore, $C^B$ simply permutes $\left(\F_2^{2n}\right)_B$ (and cannot map elements into $(\F_2^{2n})_B)$.
    Hence, 
    \[
    C^B(C^A(T)) \cap  \left(\F_2^{2n}\right)_B = C^B(C^A(T)) \cap  C^B(C^A(\left(\F_2^{2n}\right)_B) = T \cap \left( \F_2^{2n}\right)_B = \{0\}. 
    \]
    Finally, since $ \left(\F_2^{2n}\right)_{B^\prime} \subseteq  \left(\F_2^{2n}\right)_{B}$,
    \[
        C^B(C^A(T))_{B^\prime}  = C^B(C^A(T)) \cap  \left(\F_2^{2n}\right)_{B^\prime} = \{ 0\}.
    \]
    
    Because $C^A$ and $C^B$ act locally on $A$ and $B$ respectively, they do not change the entanglement between $A$ and $B$, thus completing the proof.
\end{proof}

\section{The Algorithm}\label{sec:algo}

We present and analyze our algorithm for estimating the entanglement entropy across any bipartition of qubits. 
At a high level, our algorithm computes the upper and lower bounds given in \cref{thm:main-entanglement-bounds}.  
The details are presented below in \cref{alg:main}.

\vspace{\baselineskip}
\begin{algorithm}[H]
\caption{Estimating Entanglement Entropy}\label{alg:main}
\SetKwInOut{Promise}{Promise}
\KwInput{$\frac{8 \ln(1/\delta) + 16n}{\eps^2}$ copies of $\rho = \ketbra{\psi}{\psi}$, $A \sqcup B = [n]$, $\eps \in (0, 3/8)$, and $\delta \in (0,1]$}
\Promise{$\ket\psi$ has stabilizer dimension at least $n - k$} 
\KwOutput{$\ell, u \in \R$ such that $\ell  \le \entropy(\rho_A) \le u$ and $0 \le u - \ell \le k + \max\{0, 2(\eps n + \binentropy(\eps)) - 1\}$, with probability at least $1 - \delta$}

Perform Bell difference sampling to draw $\frac{2 \ln(1/\delta) + 4n}{\eps^2}$ samples from $q_\psi$. 

Let $S$ be the symplectic complement of the subspace spanned by the samples.

Let $r \coloneqq \begin{cases}
    0 & \dim(S) = n - k,\\
    \eps n + \binentropy(\eps) & \dim(S) > n - k.
\end{cases}$

Let $u \coloneqq \min \{|A| - \dim S_A, |B| - \dim S_B\} + r$.

Let $\ell \coloneqq \max \{\dim S - \dim S_B - |A|, \dim S - \dim S_A - |B|\} - r$.

\Return{$(\ell, u)$}
    
\end{algorithm}
\vspace{\baselineskip}

Let us make a few remarks on \cref{alg:main}.
Note that the bounds $(u,\ell)$ produced by our algorithm are within a range of $k + \max\{0, 2\left(\eps n + \binentropy(\eps)\right) - 1\}$ rather than $k + 2\left(\eps n + \binentropy(\eps)\right)$ as one might na\"ively expect.
This comes from a subtle case in our analysis, the details of which are contained in the proof of \cref{thm:main-with-d}. 
We observe that in the case where our sampling procedure happens to find the exact subspace $\weyl(\ket \psi)$, we no longer need to apply \cref{lem:fannes}.
Conversely, when our sampling procedure fails to find $\weyl(\ket \psi)$ exactly, the distance between the bounds in \cref{lem:lower,lem:upper} decreases by at least $1$.

We also note that the Bell difference sampling procedure only needs to be performed once. After that, one can compute entanglement entropy bounds across any cut of qubits with only classical post-processing.

To prove the correctness of \cref{alg:main}, we use the following two statements. 
The first is about the time complexity of computing symplectic complements. 
The second says that Bell difference sampling suffices to approximately recover $\weyl(\ket\psi)$.

\begin{fact}[{\cite[Lemma 3.1]{grewal2023efficient}}]\label{fact:compute-sympcomp}
    Given a set of $m$ vectors whose span is a subspace $H \subseteq \F_2^{2n}$, there is an algorithm that outputs a basis for $H^\sympcomp$ in $O\left(mn \cdot \min(m, n)\right)$ time.
\end{fact}

\begin{lemma}[{\cite[Proof of Theorem 5.1]{grewal2023efficient}}]
\label{lem:sampling}
    Let $\ket{\psi}$ have stabilizer dimension at least $n - k$, and let $S$ be the symplectic complement of the space spanned by $\frac{2 \ln(1/\delta) + 4n}{\eps}$ samples from $q_\psi$, for some $\eps \in (0, 3/8)$. Then with probability at least $1 - \delta$, there exists a state $\ket{\phi}$ such that $S = \weyl(\ket \phi) \supseteq \weyl(\ket \psi)$ and $\abs{\braket{\psi|\phi}}^2 \ge 1 - \eps$.
\end{lemma}

We now show our main result, namely that \Cref{alg:main} is correct as specified.

 \begin{theorem}\label{thm:main-with-d}
    \Cref{alg:main} is correct and runs in time $O\left(\frac{n^3 + n^2 \log(1/\delta)}{\eps^2} \right)$.
 \end{theorem}
 \begin{proof}
    Let $S$ be the symplectic complement of the subspace spanned by our $\frac{2 \ln(1/\delta) + 4n}{\eps^2}$ samples from $q_\psi$.
    We note that these samples take $4$ copies of $\rho$ and $O(n)$ time each, and that $S$ can be computed in time $O\left(\frac{n^3 + n^2 \log(1/\delta)}{\eps^2} \right)$ by \cref{fact:compute-sympcomp}.
    
    By \Cref{lem:sampling}, with probability at least $1 - \delta$, there exists a state $\ket{\phi}$ such that $S = \weyl(\ket{\phi}) \supseteq \weyl(\ket \psi)$ and $\abs{\braket{\psi|\phi}}^2 \ge 1 - \eps^2$, and therefore $\dist_\tr(\ket{\psi}, \ket{\phi}) \le \eps$. Assume henceforth that $S$ and $\ket{\phi}$ satisfy these criteria.
    In fact, we can further assume $\dist_\tr(\ket{\psi}, \ket{\phi}) \le d$, where
    \[
    d \coloneqq \begin{cases}
        0 & \dim(S) = n - k,\\
        \eps & \dim(S) > n - k,
    \end{cases}
    \]
    because if $\dim(S) = n - k$, we must have $S = \weyl(\ket\psi)$, and hence we can choose $\ket\phi = \ket\psi$.

    Let $\sigma = \ketbra{\phi}{\phi}$. By \Cref{lem:upper}, we have
    \[
    u^\prime \coloneqq \min\{|A| - \dim S_A , |B| - \dim S_B \} \ge
    \entropy(\sigma_A).
    \]
    Similarly, \Cref{lem:lower} implies
    \[
    \ell^\prime \coloneqq \max\{\dim S - \dim S_B - |A|, \dim S - \dim S_A - |B|\} \le \entropy(\sigma_A).
    \]
    
    Note that $r = dn + \binentropy(d)$, $\ell = \ell^\prime - r$, and $u = u^\prime + r$.
    Recalling that $\dist_\tr(\ket{\psi}, \ket{\phi}) \le d$, \Cref{lem:fannes} implies that $\ell \le \entropy(\rho_A) \le  u$. So, this establishes that $u$ and $\ell$ are upper and lower bounds, respectively, on the entanglement entropy.
    
    It remains to bound the difference between $u$ and $\ell$. Observe that
    \begin{align*}
    u - \ell &= (u^\prime + r) - (\ell^\prime-r)\\
    &\le (|A| - \dim S_A) - (\dim S - \dim S_A - |B|) + 2r\\
    &= |A| + |B| - \dim S +2r \\
    &= n - \dim S +2r.
    \end{align*}
    
    In the case where $\dim S = n - k$, we have $u - \ell \le k$.
    Otherwise, when $\dim(S) > n - k$, we have 
    \begin{align*}
        n - \dim S +2r
        &\leq n - (n-k+1) + 2r\\
        &= k + 2r -1\\
        &= k + 2(\eps n + \binentropy(\eps)) - 1.
    \end{align*}
    
    In both cases, for all $S \supseteq \weyl(\ket \psi)$, we have 
    \[
    u - \ell \leq k + \max\left\{0,2r - 1\right\},
    \]
    which completes the proof.
 \end{proof}

 If we take $\eps$ to be sufficiently small, we can disregard the additional additive error $2r-1$.
 \begin{corollary}\label{cor:main-without-d}
     By setting $\eps = \frac{1}{8n}$, \cref{alg:main} outputs upper and lower bounds on the entanglement entropy $(u,\ell)$ such that $u-\ell \leq k$ with probability at least $1-\delta$. 
     It now uses $1024n^3 + 512 n^2 \ln(1/\delta)$ samples of $\ket{\psi}$ and runs in time $O\left(n^5 + n^4 \log(1/\delta) \right)$.
 \end{corollary}
\begin{proof}
    Assume $n \geq 2$, because we don't need to compute the entanglement of a single qubit state.
    \Cref{fact:binary-entropy-bound} tells us that $\binentropy(\eps) < 0.37$, and therefore $2(\eps n + H(\eps)) - 1 < 2(1/8 + 0.37) - 1 < 0$. So, $\max\left\{0, 2(\eps n + H(\eps)) - 1\right\} = 0$.
    We then appeal to \cref{thm:main-with-d}.
\end{proof}

As a corollary, we can show a lower bound on the number of non-Clifford gates necessary to prepare pseudoentangled states. 

\begin{corollary}\label{cor:pes-lowerbound}
    Any family of Clifford circuits that produces a pseudoentangled ensemble $\{\ket{\Psi_k}, \ket{\Phi_k}\}_k$ with entropy gap $f(n) \text{ vs. } g(n)$ satisfying $f(n) - g(n) \ge t$ must use $\Omega(t)$ auxiliary non-Clifford single-qubit gates.
\end{corollary}

\begin{proof}
    Suppose $t'$ non-Clifford gates are used to construct $\{\ket{\Psi_k}\}_k$ and $\{\ket{\Phi_k}\}_k$. We argue that if $2t' < t$, these state ensembles can be distinguished with non-negligible advantage in polynomial time.

    All such states $\ket{\Psi_k}$ and $\ket{\Phi_k}$ have stabilizer dimension at least $n - 2t'$, by \Cref{lem:stab_dim}. 
    The distinguisher, then, is the following: given copies of an unknown $\ket{\psi}$ belonging to one of the two ensembles, run \Cref{alg:main} according to \Cref{cor:main-without-d}, assuming stabilizer dimension at least $n - 2t'$ and $\delta = 1/3$. This produces bounds $(u, \ell)$ on the entanglement entropy of $\ket{\psi}$ across some fixed cut $(A, B)$ of size $n/2$. Then, output that $\ket{\psi} \in \{\ket{\Psi_k}\}_k$ if $\ell \le f(n) \le u$, and output $\ket{\psi} \in \{\ket{\Phi_k}\}_k$ otherwise. The algorithm guesses correctly with probability at least $2/3$, because $u - \ell \le 2t'$, so at most one of $f(n)$ and $g(n)$ can lie between $u$ and $\ell$.
\end{proof}

\section*{Acknowledgments}
We thank Tony Metger for suggesting this problem to us, and Fermi Ma for helpful conversations.

SG is supported (via Scott Aaronson) by a Vannevar Bush Fellowship from the US Department
of Defense, the Berkeley NSF-QLCI CIQC Center, a Simons Investigator Award, and the Simons “It
from Qubit” collaboration. 
VI is supported by an NSF Graduate Research Fellowship.
WK acknowledges support from the U.S.\ Department of Energy, Office of Science, National
Quantum Information Science Research Centers, Quantum Systems Accelerator. 
DL is supported by NSF award FET-2243659.

This work was done in part while SG, VI, and DL were visiting the Simons Institute for the Theory of Computing.

\bibliographystyle{alphaurl}
\bibliography{refs}

\appendix

\end{document}